\documentclass{scrartcl}
\usepackage{fullpage}
\usepackage{amssymb}
\usepackage[utf8]{inputenc}
\usepackage{amsmath,amsthm}
\usepackage{array,xcolor,colortbl,graphicx,multirow}
\usepackage{algorithm}
\usepackage[noend]{algpseudocode}
\usepackage[font={footnotesize}]{subcaption}
\usepackage[font={footnotesize}]{caption}
\usepackage{breakcites}
\usepackage{booktabs}
\usepackage{xspace}
\usepackage{hyperref}
\usepackage{graphics}
\usepackage{xcolor}
\usepackage{authblk}

\newcommand{\ALG}{\textsc{Alg}\xspace}
\newcommand{\ALGB}{\textsc{Alg}^\textrm{M}\xspace}
\newcommand{\ALGP}{\textsc{Alg}^\textrm{P}\xspace}
\newcommand{\AS}{\ensuremath{\textsc{Alg}_1}\xspace}
\newcommand{\OPT}{\textsc{Opt}\xspace}
\newcommand{\OFF}{\textsc{Off}\xspace}
\newcommand{\ONL}{\textsc{Onl}\xspace}
\newcommand{\VV}{\ensuremath{V^2}}
\newcommand{\E}{\ensuremath{\mathbf{E}}}
\newcommand{\M}{\ensuremath{M}}

\algdef{SE}[SUBALG]{Indent}{EndIndent}{}{\algorithmicend\ }
\algtext*{Indent}
\algtext*{EndIndent}

\newtheorem{lemma}{Lemma} 
\newtheorem{theorem}{Theorem}
\newtheorem{corollary}[theorem]{Corollary}

\newtoks\titleboxnotes
\newcount\titleboxnoteflag

\begin{document}
\title{Optimizing Reconfigurable Optical Datacenters: The Power of
Randomization\thanks{%
Research supported by the European Research Council (ERC), consolidator project
Self-Adjusting Networks (AdjustNet), grant agreement No.~864228, and by the
Polish National Science Centre grant no 2022/45/B/ST6/00559.
For the purpose of open access, the authors have applied CC-BY public
copyright license to any Author Accepted Manuscript (AAM) version arising from
this submission.
}}

\renewcommand{\Authfont}{\sectfont}
\renewcommand{\Affilfont}{\mdseries \itshape \small}

\author[1]{Marcin Bienkowski}
\author[2]{David Fuchssteiner}
\author[3]{Stefan Schmid}
\affil[1]{University of Wrocław, Poland}
\affil[2]{University of Vienna, Austria}
\affil[3]{Technical University of Berlin, Germany}
\date{}

\maketitle

\begin{abstract}
Reconfigurable optical topologies are a promising new technology
to improve datacenter network performance and cope with the explosive growth
of traffic.
In particular, these networks allow to directly and adaptively connect
racks between which there is currently much traffic, hence
making an optimal use of the bandwidth capacity by avoiding multi-hop forwarding.

This paper studies the dynamic optimization of such reconfigurable topologies,
by adapting the network to the traffic in an online manner. The underlying algorithmic problem can be described as an online maximum weight
$b$-matching problem, a~generalization of maximum weight matching where each
node has at most $b \geq 1$ incident matching edges. 

We make the case for a randomized approach for matching optimization. 
Our main contribution is a $O(\log b)$-competitive algorithm and we show that it is asymptotically optimal.
This algorithm is hence exponentially better than the best possible deterministic online algorithm. 

We complement our theoretical results with extensive trace-driven
simulations, based on real-world datacenter workloads.
\end{abstract}


\section{Introduction}\label{sec:intro}

Datacenter networks have become a critical infrastructure of our digital society, and the performance requirements on these networks are increasingly stringent. Indeed, especially inside datacenters, traffic is currently growing explosively, e.g., due to the popularity of data-centric applications related to AI/ML~\cite{singh2015jupiter}, but also due to the trend of resource disaggregation (e.g., requiring fast access to remote resources such as GPUs), as well as hardware-driven and distributed
training~\cite{li2019hpcc}.
Accordingly, over the last years, great efforts have been made to improve datacenter networks~\cite{kellerer2019adaptable}. 

A particularly innovative approach to improve datacenter network performance, is to dynamically adjust the datacenter topologies towards the workload they serve, in a \emph{dynamic} and \emph{demand-aware} manner. Such adjustments are enabled by emerging reconfigurable optical communication technologies, such as optical circuit switches, which provide dynamic \emph{matchings} between racks~\cite{sirius,zhou2012mirror, zerwas2023duo,addanki2023mars, rotornet, opera, helios, firefly, megaswitch, quartz, chen2014osa,projector,cthrough,splaynet,venkatakrishnan2018costly,schwartz2019online,proteus,osa,100times,fleet,osn21,griner2021cerberus}. 
Indeed, empirical studies have shown that datacenter traffic features much spatial and temporal structure~\cite{sigmetrics20complexity,benson2010network,roy2015inside,foerster2023analyzing-www}, which may be exploited for optimization.

This paper studies the optimization problem underlying such reconfigurable datacenter networks. In particular, we consider a typical leaf-spine datacenter network where a set of racks are interconnected by $b$ optical circuit switches, each of which provides one matching between top-of-rack switches, so $b$ matchings in total.
In a nutshell, the goal is to optimize these matchings so that the link resources are used optimally, i.e., the number of hops taken per communicated bit is minimized (more details will follow). 

\subsection{The model}

Our problem optimizing an optical reconfigurable datacenter network can be modeled as an~online dynamic version of the 	classic
$b$-matching problem~\cite{anstee1987polynomial}. In this problem, each node can be connected with at 
most $b$ other nodes (using optical links), which results in a $b$-matching.

\paragraph{Input.}

We are given an arbitrary (undirected) static weighted
and connected network on the set of nodes $V$ (i.e., the top-of-rack switches) connected by a~set of
non-configurable links $F$: the fixed network (based, e.g., on a Clos or fat-tree topology). 
Let $\VV$ be the set of
all possible unordered pairs of nodes from~$V$. For any node pair 
$e = \{s,t\} \in \VV$, we call $s$ and $t$ the \emph{endpoints} of~$e$,
and we let $\ell_e$ denote the length of a
shortest path between nodes $u$ and~$v$ in graph $G=(V,F)$.
Note that $u$ and $v$
are not necessarily directly connected in~$F$.

The fixed network can be enhanced with reconfigurable links, providing a
matching of degree~$b$: Any node pair from~$\VV$ may become a \emph{matching
edge} (such an edge corresponds to a~reconfigurable optical link), but the
number of matching edges incident to any node has to be at most~$b$, for a given
integer $b \geq 1$. 

The demand is modeled as a sequence of 
communication requests\footnote{A request could either be an individual packet or a certain amount of bytes transferred.
This model of a~request sequence is often considered in the literature and is more fine-grained than, e.g., a sequence of traffic matrices.} 
$\sigma=\{s_1,t_1\},\{s_2,t_2\}, \ldots$ revealed over time,
where $\{s_i,t_i\} \in V^2$.

\paragraph{Output and objective.} 

The task is to 
schedule the reconfigurable links over time, that is,
to maintain a dynamically
changing $b$-matching $M \subseteq \VV$. This means that each node pair from~$M$ is called 
a \emph{matching edge}, and we require that each node has
at most $b$ incident matching edges. We aim to jointly minimize routing
and reconfiguration costs, defined below.

\paragraph{Costs.}

We use the cost of routing a request as our cost measurement. 
This measurement is directly related to the throughput in the network, because routing can be seen as a form of bandwidth tax~\cite{rotornet,griner2021cerberus}.
Furthermore, analytical results show that the throughput of a network is inversely proportional to the route length~\cite{addanki2023mars, namyar2021throughput}.
Empirical results also show that lower routing costs have a benefit on an application's execution time~\cite{griner2021cerberus}.
In our model, the cost of serving (i.e., routing cost) a request $e = \{s,t\}$ depends on
whether $s$ and $t$ are connected by a~matching edge. A~given
request can either only take the fixed network or a direct matching edge (i.e.,
routing is segregated~\cite{projector,ancs18}). If $e \notin M$, the requests are routed
exclusively on the fixed network, and the corresponding cost is~$\ell_e$.
(shorter paths imply smaller resource costs, i.e., lower ``bandwidth
tax''~\cite{rotornet,griner2021cerberus}). If $e \in M$, the request is served by the matching
edge, and the routing costs 1. 

Once the request is served, an algorithm may modify the set of matching edges: 
reconfiguration costs $\alpha$ per each node pair added or removed from the
matching $M$. (The reconfiguration cost and time can be assumed to be independent
of the specific edge.)

\paragraph{Online algorithms.} 

A (randomized) algorithm $\ONL$ is \emph{online}
if it has to take decisions without knowing the future requests
(in our case, e.g., which edge to include next in the matching
and which to evict). 
Such an algorithm is said to be $\rho$-competitive~\cite{BorEl-98} 
if there exists $\beta$ such that 
for any input instance~$\sigma$, it holds that 
\[
	\E[\ONL(\sigma)] \leq \rho \cdot \OPT(\sigma) + \beta \,,
\]
where $\OPT(\sigma)$ is the cost of the optimal (offline) solution 
for~$\sigma$ and $\E[\ONL(\sigma)]$ is the 
expected cost of algorithm \ONL on~$\sigma$. The expectation
is taken over all random choices of $\ONL$, i.e., the input itself
is worst-possible, created adversarially.
It is worth noting that $\beta$ can depend on 
the parameters of the network, such as the number of nodes, 
but has to be independent of the actual sequence of requests. 
Hence, in the long run, this additive term $\beta$ becomes negligible
in comparison to the actual cost of online algorithm $\ONL$.

\paragraph{Generalization to (b,a)-matching problem.} 

The comparison of an online algorithm to the fully clairvoyant 
offline solution may seem unfair. 
Therefore, we consider our $b$-matching problem also in a more generalized setting. 
We denote this generalization as \emph{$(b,a)$-matching}, where the restrictions 
imposed on an online algorithm remain unchanged: the number of 
matching edges incident to any node has to be at most $b$. However, 
the optimal solution is more constrained: the number of 
matching edges has to be at most $a \leq b$ for any node. 
Similar settings are studied frequently in the literature, as \emph{resource augmentation} models,
e.g., in the context of paging and 
caching~\cite{SleTar85,Young91,Young94,BansalBN12} or scheduling problems, see, e.g.,
\cite{ChadhaGKM09,KalyanasundaramP00,spaa21rdcn}.

\subsection{Our contributions}

Motivated by the problem of how to establish topological shortcuts
in datacenter networks supported by $b$ optical switches, we 
present a randomized online algorithm for the 
$(b,a)$-matching problem. 

Our proposed solution achieves a competitive ratio of 
$O((1 + \max_e \{ \ell_e \} / \alpha) \cdot \log (b/(b-a+1))$, which is almost optimal: we will also show a lower bound of
$\Omega(\log (b/(b-a+1))$. 
We note that in all practical applications $\alpha$ is by several orders of 
magnitude greater than $\ell_{\max}$, and thus the factor $1 + \max_e \{ \ell_e \} / \alpha$ is close to $1$. 

When our algorithm is compared to an optimal algorithm that has to maintain
$b$-matching as well (i.e., when $a=b$), the competitive ratio of our algorithm 
becomes $O((1 + \max_e \{ \ell_e \} / \alpha) \cdot \log b)$ and the lower bound becomes $\Omega(\log b)$.
It is worth noting that the best deterministic online algorithm for this problem
is only $\Theta(b)$-competitive. 

Our analysis relies on a reduction to a
uniform case (where $\alpha = 1$ and all path lengths are equal to $1$), 
which allows us to avoid delicate charging arguments
and enables a simplified analytical treatment. 

We show that our randomized algorithm is not only better with respect to the worst-case competitive ratio than the deterministic online algorithm in theory,
but also attractive in practice. 
Our empirical evaluation, based on various real-world datacenter traces,
shows that our algorithm is significantly faster while achieving roughly
the same matching quality.

As a contribution to the research community, to facilitate follow-up work, and to 
ensure reproducibility, we open-source our implementation and all experimental artifacts
together with this paper. 

\subsection{Organization}

The remainder of this paper is organized as follows.
We present our optimization framework together with a formal analysis in \S\ref{sec:algorithm}.
\S\ref{sec:simulations} reports on our empirical results.
After reviewing related work in \S\ref{sec:relwork}, we conclude our contribution
and discuss future work in \S\ref{sec:conclusion}.

\section{Algorithm and Analysis}\label{sec:algorithm}

Recall that our goal is to adapt the reconfigurable datacenter links such that bandwidth resources are used optimally (i.e., bits travel the least number of hops) while accounting for reconfiguration costs, in an online manner. To this end, we aim to compute heavy matchings between nodes, e.g., the racks resp.~top-of-rack switches in the reconfigurable datacenter. 

We first define the \emph{uniform} variant of the $(b,a)$-matching. There, 
the distance between each pair of nodes is~$1$ (i.e., $\ell_e = 1$ for any $e \in \VV$)
and $\alpha = 1$ (recall that $\alpha$ is the cost of adding or removing a matching edge to/from $M$).

In the following, let 
\begin{align*}
    \ell_{\max} = \max_e \{ \ell_e \} && \text{and} && \gamma = 1 + \ell_{\max} / \alpha.
\end{align*}
We first show that it suffices to solve the uniform variant of the problem:
once we do this, we can get an~algorithm for the general variant, losing an additional factor 
of~$O(\gamma)$. 

Afterwards, we show a simple algorithm that uses known randomized algorithms
for the paging problem to solve the uniform variant of the $(b,a)$-matching. This will yield 
a randomized $O(\gamma \cdot \log (b/(b-a+1))$-competitive algorithm for the $(b,a)$-matching problem.

\subsection{Reduction to uniform case}
\label{sec:reduction}

We now show a reduction to the uniform case.

\begin{theorem}
\label{thm:reduction}
Assume there exists a (deterministic or randomized) $c$-competitive 
algorithm \AS for the variant of the $(b,a)$-matching problem for the uniform case.
Then there exists a (respectively, deterministic or randomized)~$4 \gamma \cdot c$-competitive 
algorithm \ALG for the $(b,a)$-matching problem.
\end{theorem}

\begin{proof}
Let~$I$ be the input for algorithm \ALG. \ALG creates (in an online manner) another
input instance~$I_1$.
For any node pair~$e$ let 
\[ 
    k_e = \lceil \alpha / \ell_e \rceil.
\]
For any node pair~$e$ independently, we call each~$k_e$-th
request to~$e$ in~$I$ \emph{special}; the remaining ones are 
called \emph{ordinary}. Input~$I_1$ contains only special requests from~$I$.
Furthermore, for~$I_1$, we assume that $\alpha = 1$ and the distance 
between each pair of nodes is~$1$. \ALG simply runs
$\AS$ on~$I_1$ and repeats its reconfiguration choices (modifications of
the matching $\M$). In particular, this means that $\ALG$ may performs changes to~$M$ 
only upon the~$k_e$-th occurrence of a given node pair~$e$. 

To prove the theorem, we will show the following inequalities 
for some value~$\beta$ independent of the input~$I$.
\begin{enumerate}
\item $\ALG(I) \leq 2 \gamma \cdot \alpha \cdot \AS(I_1) + |V^2| \cdot \gamma \cdot \alpha$
\item $\AS(I_1) \leq c \cdot \OPT(I_1) + \beta$
\item $\OPT(I_1) \leq (2 / \alpha) \cdot \OPT(I)$
\end{enumerate}

\paragraph{Showing the first inequality.}

We fix any node pair~$e$ and we analyze 
the cost pertaining to handling~$e$ both by \ALG and~$\AS$. 
We partition~$I$ into disjoint intervals, so that the first
interval starts at the beginning of the input sequence, and 
each interval except the last one ends at the special request to~$e$ inclusively.
(We note that the partition differs for different choices of $e$.)
We look at any non-last interval~$S$ and its counterpart~$S_1$ in the input~$I_1$. Note that except~$k_e$ requests to~$e$ and one request
to~$e$ in~$S$, these intervals may contain also requests to other node pairs. We now
compare the costs pertaining to~$e$ in~$S$ incurred on \ALG (denoted $\ALG(S,e)$), 
to the costs pertaining to~$e$ in~$S_1$ incurred on~$\AS$ (denoted $\AS(S_1,e)$).
We consider two cases.

\begin{itemize}
\item 
If $\AS(S_1,e) = 0$, then 
$\AS$ must have~$e$ in~$\M$ when it is requested at the end of~$S_1$, and 
moreover, it cannot perform any reconfigurations that touch~$e$. This means that it must 
have~$e$ in~$M$ right after the special request to~$e$ preceding~$S_1$, and 
keep~$e$ in~$\M$ throughout the whole considered interval~$S_1$.
As $\ALG$ is mimicking the choices of $\AS$, it has~$e$ in $\M$ during~$S$, and thus all~$k_e$ requests to~$e$ in~$S$
are free for \ALG and $\ALG(S,e) = 0$ as well.

\item Otherwise, let $q = \AS(S_1,e) \geq 1$. As $\AS$ performed at most $q$
reconfigurations concerning $e$ in $S_1$, so does \ALG in~$S$. Furthermore, \ALG pays at
most for~$k_e$ requests to~$e$ ($\ell_e$ for each). 
We note that 
\begin{align*} 
    k_e \cdot \ell_e & < (\alpha / \ell_e + 1) \cdot \ell_e = \alpha + \ell_e \\
    & \leq \alpha + \ell_{\max} = \gamma \cdot \alpha,
\end{align*}
and thus $\ALG(S,e) \leq q \cdot \alpha + k_e \cdot \ell_e 
\leq (q + \gamma) \cdot \alpha < 2 q \gamma \cdot \alpha 
= 2 \gamma \cdot \alpha \cdot \AS(S_1,e)$.
\end{itemize}

The first inequality follows by summing over all (non-last) intervals 
and all possible choices of~$e$. The term 
$|\VV| \cdot \gamma \cdot \alpha$ upper-bounds the total cost in the last 
intervals: there are $|\VV|$ of them, and in each the cost is at most~$k_e \cdot \ell_e < \gamma \cdot \alpha$.

\paragraph{Showing the second inequality.}

This one follows immediately as $\AS$ is $c$-competitive on input~$I_1$.

\paragraph{Showing the third inequality.}

We again fix any edge~$e$ and consider the same
partitioning into intervals as above. This time, however, on the basis of
solution $\OPT(I)$, we construct an offline solution $\OFF$ for the input~$I_1$
by mimicking all reconfiguration choices of $\OPT$ on input~$I_1$. Note that~$I_1$ contains only a subset of requests from~$I$ and thus, in response to a
single request in~$I_1$, $\OFF$ may react with a sequence of reconfigurations
that are redundant (i.e., remove some edge from $\M$ and then fetch it back).
However, as the matching~$\M$ that \OFF maintains is the same as that of \OPT,
it is feasible.

We now argue that for any interval~$S$ in~$I$ and the corresponding interval~$S_1$ in~$I_1$,
it holds that $\OFF(S_1,e) \leq (2 / \alpha) \cdot \OPT(S,e)$. 

\begin{itemize}
\item If $\OPT(S,e) < \alpha$, then $\OPT$ performs no reconfigurations concerning $e$.
In this case $\OPT$ has to pay for all requests to~$e$ within~$S$ or none of them.
The first case would imply that its cost is at least~$k_e \cdot \ell_e \geq \alpha$,
and thus the only possibility is that it does not pay for any request. 
In this case, \OFF does not perform any reorganizations pertaining to~$e$ either and
does not pay for the only request to~$e$ in~$S_1$. Hence, $\OFF(S_1,e) = 0 =
\OPT(S,e)$. 

\item Otherwise, let $q = \OPT(S,e) \geq \alpha$. $\OPT$ performs at most
$\lfloor q / \alpha \rfloor$ reconfigurations pertaining to~$e$. $\OFF$
executes the same reconfigurations, and hence it pays at most $\lfloor q / \alpha \rfloor$ for
reconfigurations in~$S_1$ pertaining to~$e$ and at most~$1$ for the only request
to~$e$ in $S_1$. Thus, $\OFF(S_1,e) \leq \lfloor q / \alpha \rfloor + 1 \leq 2 \cdot (q /
\alpha) = (2 / \alpha) \cdot \OPT(S,e)$.
\end{itemize}

The third inequality now follows by summing 
the relation $\OFF(S_1,e) \leq (2 / \alpha) \cdot \OPT(S,e)$
over all intervals (including the last ones) and node pairs~$e$, 
and observing that the optimal cost for~$I_1$ can be only smaller than the cost of $\OFF$ on~$I_1$.

\paragraph{Combining all three inequalities.}

By combining all three inequalities, we obtain that 
\begin{align*}
    \ALG(I) 
    & \leq 2 \gamma  \cdot \alpha \cdot \AS(I_1) + |V^2| \cdot \gamma \cdot \alpha \\ 
    & \leq 2 \gamma \cdot \alpha \cdot c \cdot \OPT(I_1) + 2 \gamma \cdot \alpha \cdot \beta
        + |V^2| \cdot \gamma \cdot \alpha \\
    & \leq 4 \gamma \cdot c \cdot \OPT(I) + (|V^2| + 2\beta) \cdot \gamma \cdot \alpha,
\end{align*}
which concludes the proof.
\end{proof}

\subsection{Algorithm for uniform case}
\label{sec:alg-uniform}

Below we present a randomized algorithm \AS, which solves the uniform variant of the 
problem (where $\ell_e = 1$ for all $e \in \VV$ and $\alpha = 1$). 

Let the $(b,a)$-paging problem be a resource-augmented variant of the paging problem~\cite{SleTar85}
where an online algorithm has cache of size~$b$ and optimal algorithm has 
a cache of size~$a$. 

\begin{theorem}
\label{thm:paging_reduction}
Assume there exists a (deterministic or randomized) 
$f(b,a)$-competitive algorithm for the $(b,a)$-paging problem
for some function $f: \mathbb{N}_{\geq 0} \times \mathbb{N}_{\geq 0} \to \mathbb{R}_\geq 0$.
Then, there exists an~(also deterministic or randomized) $O(f(b,a))$-competitive algorithm 
for the uniform variant of the $(b,a)$-matching problem.
\end{theorem}

We note that the cost model in the paging problem as defined in many papers (see, e.g.,
\cite{SleTar85,FKLMSY91,Young91}) differs slightly from ours in two aspects:
\begin{itemize}
\item In the paging problem, whenever a page is requested, it must be fetched 
to the cache if it is not yet there (bypassing is not allowed). 
In contrast, in the $(b,a)$-matching problem an algorithm does 
not have to include the requested node pair in the matching. 
\item In the paging problem, an algorithm pays only for including page in the cache (there is no cost for 
eviction as in our model). Because bypassing is not allowed, the paging problem 
does not include the cost of serving a request either. 
\end{itemize}
We will handle these differences in our proof.

\paragraph{Algorithm definition.} 

Our algorithm $\AS$ for the $(b,a)$-matching problem runs separate $(b,a)$-paging
algorithms for each node; initially the caches corresponding to all vertices
are empty. $\AS$ dynamically creates input sequences for the paging algorithms
running at particular vertices. 
An~input sequence for the algorithm running at node~$w$ is
a sequence of node pairs having~$w$ as one of their endpoints.
At all times, the paging algorithm keeps a subset of at most~$b$ such
node pairs in its cache. 
On the basis of the (possibly random) decisions at each
node, \AS constructs its own solution, choosing which node pairs are kept as
matching edges in $\M$.

More precisely, whenever \AS handles a request $e = (u,v)$, it passes query~$e$ 
to the paging algorithms running at its endpoints: separately to node~$u$ and to 
node~$v$. By the
definition of the paging problem, both these vertices may reorganize their
caches, removing an arbitrary number of elements (node pairs) from their caches,  and
afterwards they need to fetch~$e$ to their caches (if it is not already there).

On this basis, \AS reorganizes the matching maintaining the following invariant:

\begin{quote}
\emph{Any node pair~$q$ is kept in the matching if and only if~$q$ is in the caches of both its endpoints.}
\end{quote}

Therefore, when~$e = (u,v)$ is requested, the actions of \AS are limited to the following: 
(i)~some node pairs having one of~$u$ or~$v$ as endpoints may be removed from~$M$, 
(ii)~$e$ becomes a matching edge in~$M$ 
(if it was not already there).\footnote{Note that there are cases 
where algorithm \AS has less matching edges than allowed by the threshold~$b$.
While this does not hinder theoretical analysis, it is worth noting 
that having an edge in the matching can only help us. Thus, 
in our experiments the removals are lazy, i.e., 
an edge is marked for removal and then some marked edges are pruned whenever
their number incident to a node exceeds~$b$.}
We note that our reduction itself is purely deterministic, but results in a randomized algorithm
if we use randomized algorithms for solving paging sub-problems at vertices.

\paragraph{Bounding competitive ratio.} 

Now we proceed with showing the desired competitive ratio of $\AS$. 

\begin{proof}[Proof of \autoref{thm:paging_reduction}]

Fix any input instance~$I$ for the $(b,a)$-matching problem. 
It induces~$|V|$ paging instances for each node: we denote the instance at node~$v$ 
by~$I_v$. Let~$A_v$ be the instance of online paging algorithm 
run at~$v$. By the theorem assumptions, for any node~$v$ it holds that 
\begin{equation}
\label{eq:paging-is-competitive}
    A_v(I_v) \leq f(b,a) \cdot \OPT(I_v) + \beta_{b,a}
\end{equation}
where $\beta_{b,a}$ is a constant independent of the input sequence and $\OPT(I_v)$ denotes 
the optimal 										$a$-paging solution for input~$I_v$.

We will show the following relations.
\begin{itemize}
    \item $\AS(I) \leq 4 \cdot \sum_{v \in V} A_v(I_v)$.
    \item $\sum_{v \in V} \OPT(I_v) \leq 14 \cdot \OPT(I)$
\end{itemize}
We note that we do not aim at optimizing the constants here, but rather 
at the simplicity of the description.
Clearly, combining these two inequalities 
with \eqref{eq:paging-is-competitive} immediately yields the theorem.

We start with proving the second inequality.
On the basis of the optimal solution for the $a$-matching problem $\OPT(I)$,
we first create an offline solution $\OFF(I)$ for the same problem, but having the 
property that right after a node pair is requested it is included in the matching.
We call this \emph{forcing property}.
This can be achieved by changing $\OPT$ decisions in the following way:
whenever we have a request~$e = (u,v)$ such that~$e$ is not in the matching 
neither directly before or after serving the request, we modify the solution 
in the following way: after executing $\OPT$ reorganizations (if any), we include~$e$ in the matching. If such a modification causes the degree of matching at~$u$ to exceed~$b$, we evict an arbitrary edge $e_u \neq e$ incident 
to~$u$. We perform an analogous action for~$v$, evicting edge~$e_v$ if necessary. 
Afterwards, right after serving the next request, but before implementing 
reorganizations of \OPT, we revert these changes: remove~$e$ from~$M$ 
and include~$e_u$ and~$e_v$ back into~$M$.
Note that $\OFF(I) \leq 7 \cdot \OPT(I)$ as for a request to~$e$ for which 
$\OPT$ paid~$1$, the solution of $\OFF$ adds at most $3+3$ updates of the matching~$M$
at the total cost of $6 \cdot \alpha = 6$.

Now, we observe that for any node~$v$, the solution $\OFF(I)$ naturally induces
feasible solutions $\OFF_v(I_v)$ for paging problem inputs~$I_v$: these solutions
simply repeat all actions of \OFF pertaining to edges whose one endpoint is equal 
to~$v$. By the forcing property of $\OFF$, the solutions 
$\OFF_v$ are feasible: upon request to node pair $e = (u,v)$, it is 
fetched to the cache of~$v$. 
Furthermore, $\sum_{v \in V} \OFF_v(I_v) \leq 2 \cdot \OFF(I)$,
as inclusion of $e = (u,v)$ to~$M$ in the solution of $\OFF$ leads to two fetches 
in the solutions of $\OFF_u$ and $\OFF_v$.
(A more careful analysis including the evictions costs that are not present in 
the paging problem would show that $\sum_{v \in V} \OFF_v(I_v)$ 
and $\OFF(I)$ can differ at most by an additive constant independent of the input).

Finally, we observe that an optimal solution for~$I_v$ can only be cheaper 
than $\OFF(I)$, and therefore $\sum_{v \in V} \OPT(I_v)
\leq \sum_{v \in V} \OFF(I_v) \leq 2 \cdot \OFF(I) \leq 2 \cdot 7 \cdot \OPT(I)$.

We now proceed to prove the first inequality.
When a requested node pair $e = (u,v)$ is in matching~$M$, then \AS does nothing 
and no cost is incurred on $\AS$ or on any of algorithms~$A_v$; hence,
we may assume that~$e \notin M$.
In such a case, 
either~$e$ is not in the cache of~$u$, or~$e$ is not in the cache of~$v$,
or~$e$ is in neither of these two caches. That is, either~$u$ or~$v$, or both 
must fetch~$e$ to their caches, which causes $A_u(I_u) + A_v(I_v)$ to grow by 
$1$ or $2$. Moreover, paging algorithms running at~$u$ and~$v$ may evict 
some number of node pairs from their caches (these actions are free in the paging
problem). In effect, \AS pays for the request~$e$ (at cost~1), 
includes~$e$ in the matching~$M$ (at cost~$1$) and 
removes some number of edges from~$M$.

Let $\AS^+(I)$ be the cost of \AS on~$I$ neglecting the cost of removals
of edges from~$M$. The analysis above showed that increases of $\AS^+(I)$ by 
$2$ can be charged to the increases of $\sum_{v \in V} A_v(I_v)$ by at least $1$,
and thus $\AS^+(I) \leq 2 \cdot \sum_{v \in V} A_v(I_v)$. 
The proof of the first inequality follows by noting that 
the total number of removals from~$M$ is at most the total number of inclusions 
to~$M$, and thus 
$\AS(I) \leq 2 \cdot \AS^+(I) \leq 4 \cdot \sum_{v \in V} A_v(I_v)$. 
\end{proof}

\subsection{Competitive ratio}

\begin{corollary}
There exists an $O(\gamma \cdot \log (b/(b-a+1))$-competitive randomized algorithm R-BMA for 
the $(b,a)$-matching problem, for an arbitrary $\alpha \geq 1$ and arbitrary path lengths~$\ell_e$.
\end{corollary}

\begin{proof}
Applying a $2 \cdot \ln (b/(b-a+1))$-competitive randomized algorithm for the paging 
problem~\cite{Young91} (better constant factors were achieved 
when $a = b$~\cite{FKLMSY91,McGSle91,AcChNo00}) to 
\autoref{thm:paging_reduction} 
yields an $O(\log (b/(b-a+1))$-competitive 
randomized algorithm for the uniform variant of the $(b,a)$-matching problem.

Next, \autoref{thm:reduction} shows how to create an 
$O(\gamma \cdot \log (b/(b-a+1))$-competitive randomized algorithm R-BMA for 
arbitrary $\alpha \geq 1$ and arbitrary path lengths $\ell_e$.
\end{proof}

\subsection{Lower bound}

We now show that our algorithm is asymptotically optimal by proving that the 
$(b,a)$-matching problem contains the $(b,a)$-paging problem with bypassing as a special case.
In the bypassing variant of the $(b,a)$-paging problem, an algorithm does not have to 
fetch the requested item to the cache. 

\begin{lemma}
\label{lem:lower_bound}
Assume that there exists a (deterministic or randomized) 
$f(b,a)$-competitive algorithm for the $(b,a)$-matching problem for 
some function $f: \mathbb{N}_{\geq 0} \times \mathbb{N}_{\geq 0} \to \mathbb{R}_\geq 0$.
Then, there exists an (also deterministic or randomized) $4 \cdot f(b,a)$-competitive algorithm 
for the $(b,a)$-paging with bypassing.
\end{lemma}

\begin{proof}
Let $\ALGB$ be an algorithm for the $(b,a)$-matching problem. 
Assume that an algorithm for $(b,a)$-paging with bypassing has to operate in a universe of $n > b$ items. 
To construct an algorithm for the paging problem, we thus create a star graph of~$n+1$ nodes $\{v_0, v_1, \dots, v_n\}$ 
and set of~$n$ non-configurable links~$F$ connecting~$v_0$ with all remaining nodes, each of length~$1$. 
Nodes~$v_1, \dots, v_n$ correspond to the universe of~$n$ items.  For any paging request to an item~$v_i$, 
$\ALGP$ generates a block of $\alpha$ requests to node pair $\{v_0, v_i\}$. 
$\ALGP$ internally runs $\ALGB$ on the star graph 
and repeats its choices: $\ALGP$ always caches the items that are connected by the matching edges to~$v_0$ in
the solution of $\ALGB$.

For now, we ignore the costs of removing edges from the matching.
It is easy to observe that the cost of $\ALGP$ is at most $2 / \alpha$ times larger than that of $\ALGB$.
Furthermore, without loss of generality, for a block of requests to node pair $\{v_0, v_i\}$ 
an optimal algorithm for the $(b,a)$-matching either includes $\{v_0, v_i\}$ in the matching right before the block or
it does not change its matching at all. Thus the optimal solutions for the $(b,a)$-matching problem 
and $(b,a)$-paging problem coincide and their costs 
differ exactly by the factor of $1/\alpha$. Putting these bounds together,
we obtain that $\ALGP$ is $2 \cdot f(b,a)$-competitive for the $(b,a)$-paging problem with bypassing,
ignoring the costs of matching removals, and thus at most $4 \cdot f(b,a)$-competitive when matching removals
are taken into account. 
\end{proof}

As noted by Epstein et al.~\cite{epstein11} the bypassing 
variant is asymptotically equivalent to the non-bypassing one. Using the known lower bound for the $(b,a)$-paging 
problem~\cite{Young91}
along with \autoref{lem:lower_bound}, we immediately obtain the following corollary.

\begin{theorem}
The competitive ratio of any randomized algorithm for the $(b,a)$-matching problem 
is at least $\Omega(\log (b/(b-a+1)))$. The results hold for an arbitrary $\alpha$.
\end{theorem}

\begin{figure*}[!htb]
\centering
\begin{subfigure}[b]{0.325\textwidth}
		\centering
		\includegraphics[width=0.99\textwidth]{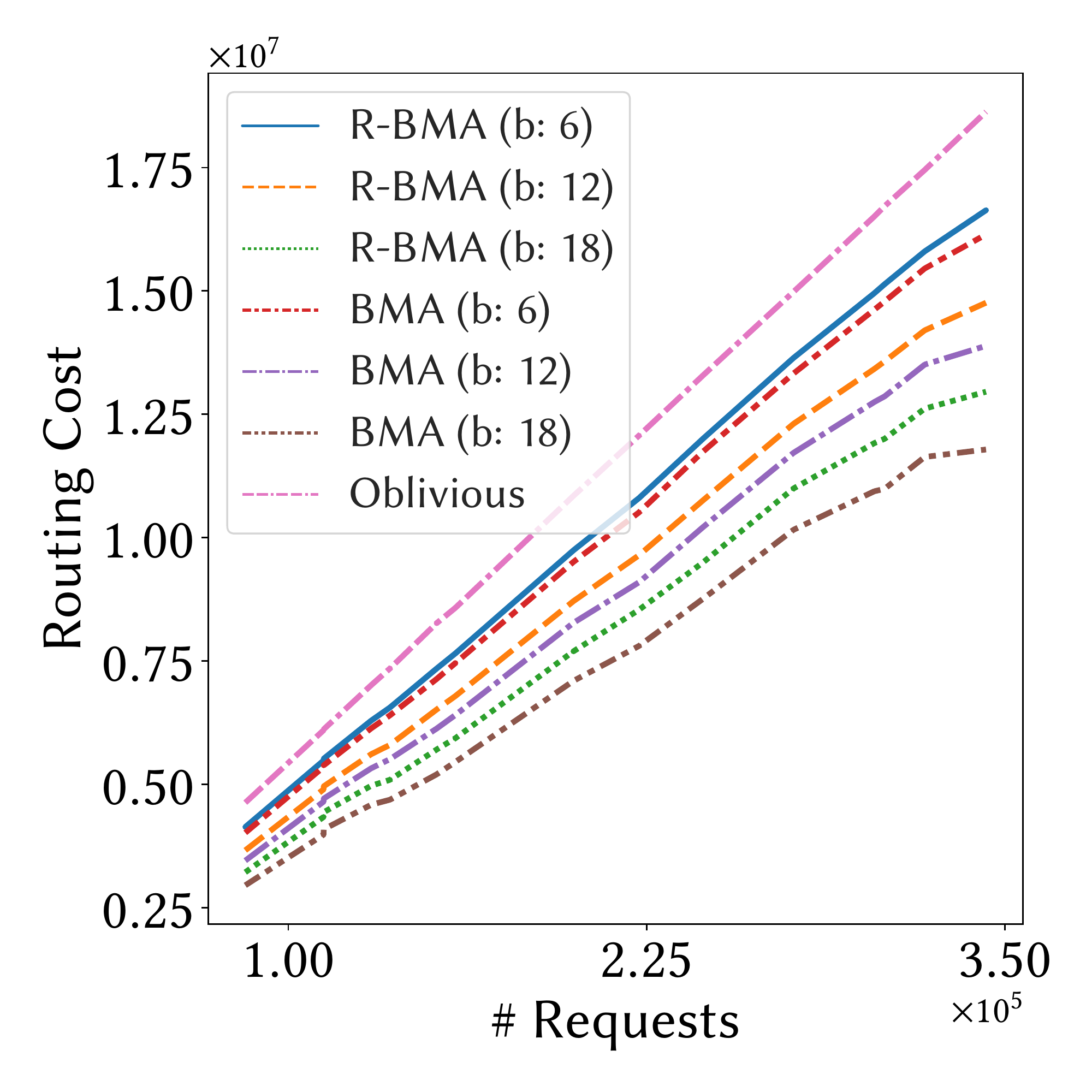}
		\caption{Routing costs.}
		\label{subfig:fb-database-results-obj}
\end{subfigure}
\hfill
\begin{subfigure}[b]{0.325\textwidth}
		\centering
		\includegraphics[width=0.99\textwidth]{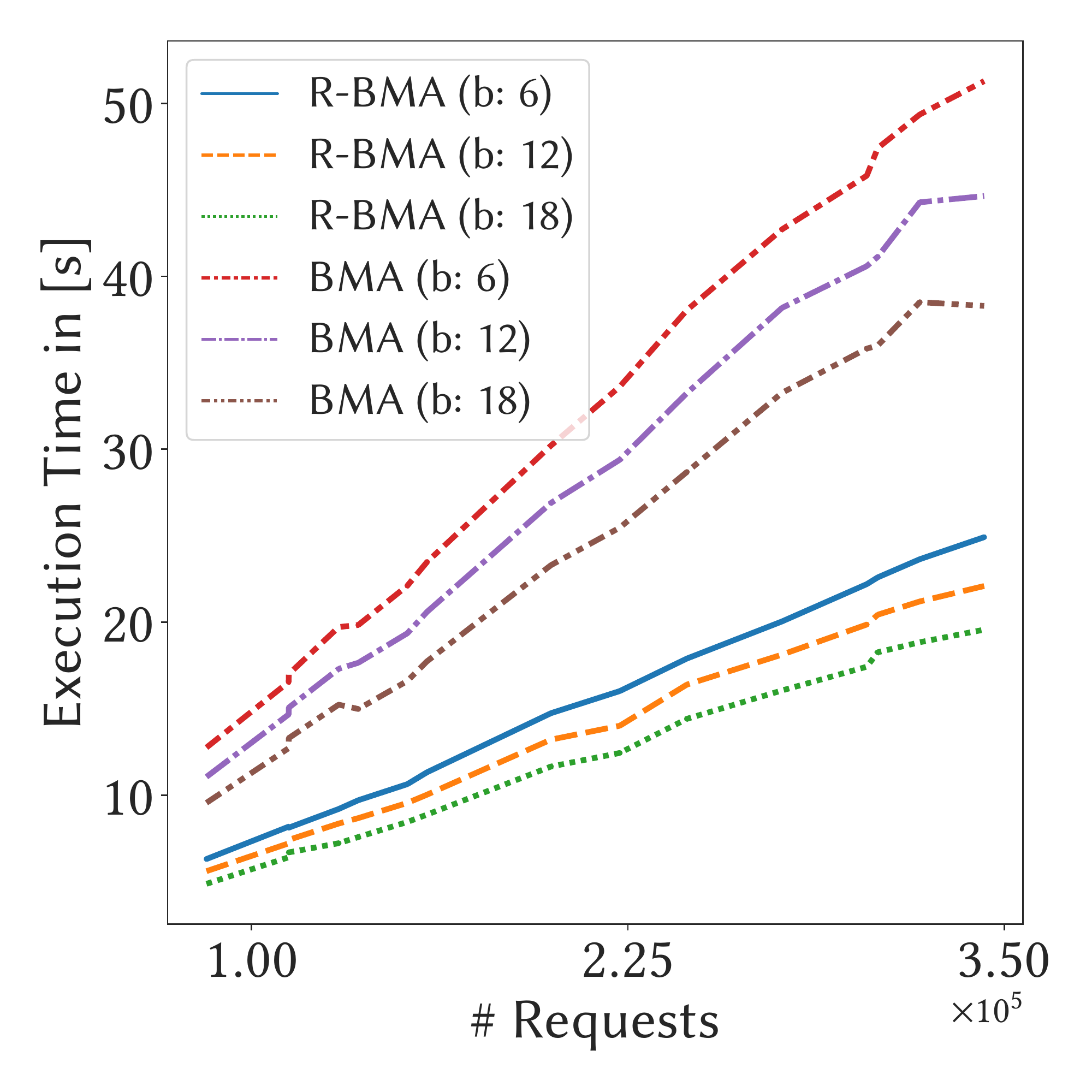}
		\caption{Execution time.}
		\label{subfig:fb-database-results-exe}
\end{subfigure}
\hfill
\begin{subfigure}[b]{0.325\textwidth}
		\centering
		\includegraphics[width=0.99\textwidth]{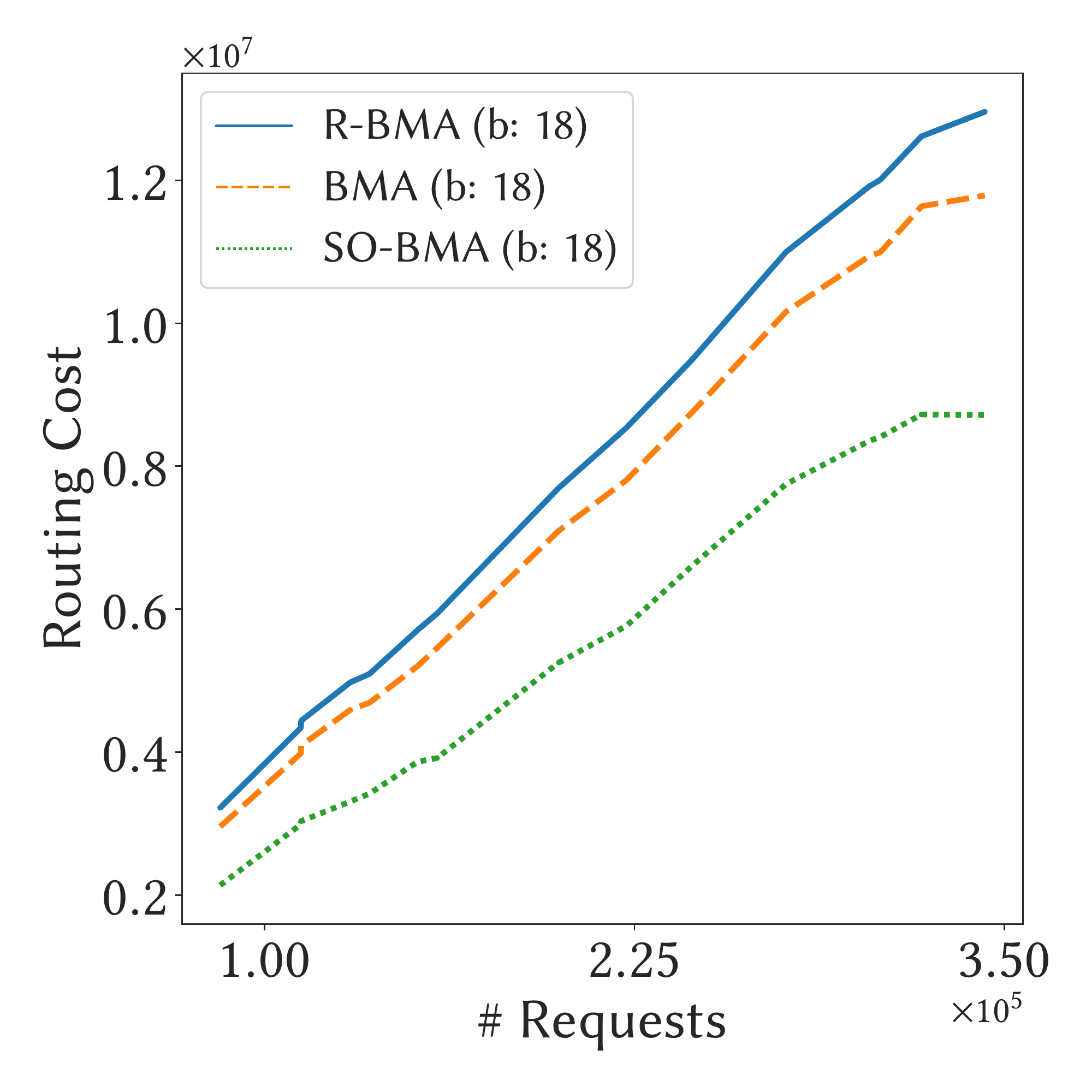}
		\caption{Best of comparison.}
		\label{subfig:fb-database-results-best}
\end{subfigure}
\caption{Facebook Database cluster.}
\label{fig:fb-database-results}
\end{figure*}

\begin{figure*}[!htb]
\centering
\begin{subfigure}[b]{0.325\textwidth}
		\centering
		\includegraphics[width=0.99\textwidth]{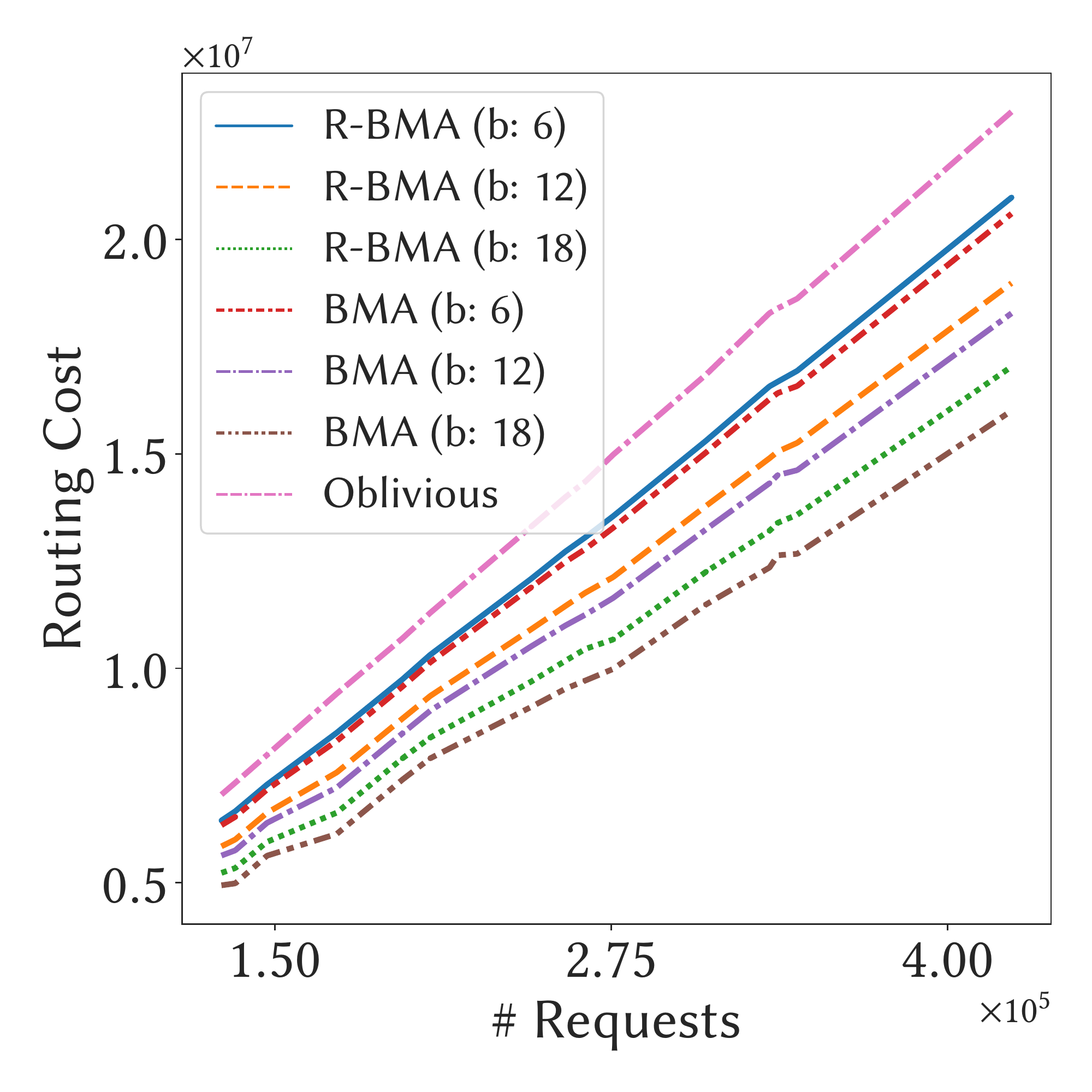}
		\caption{Routing costs.}
		\label{subfig:fb-webservice-results-obj}
\end{subfigure}
\hfill
\begin{subfigure}[b]{0.325\textwidth}
		\centering
		\includegraphics[width=0.99\textwidth]{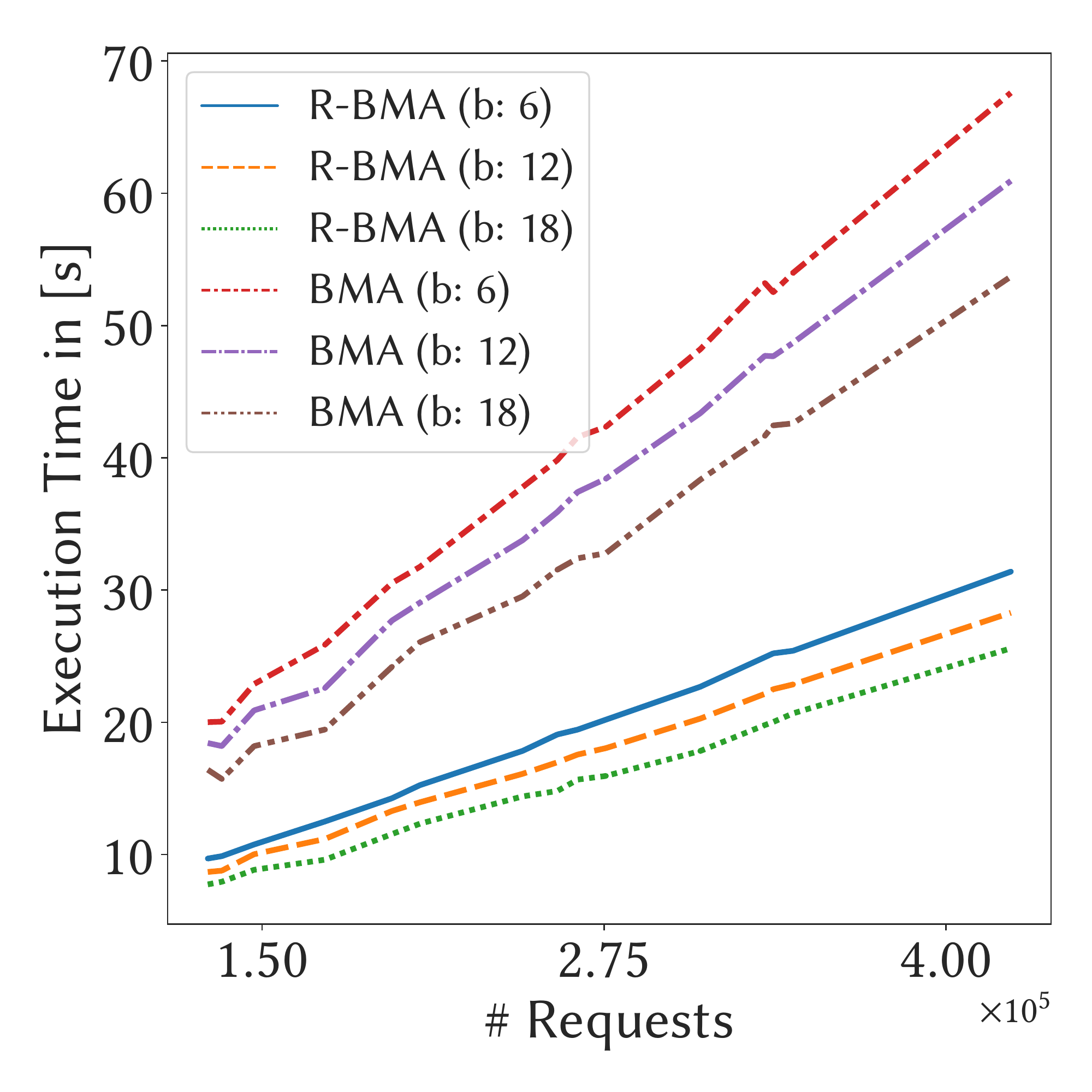}
		\caption{Execution time.}
		\label{subfig:fb-webservice-results-exe}
\end{subfigure}
\hfill
\begin{subfigure}[b]{0.325\textwidth}
		\centering
		\includegraphics[width=0.99\textwidth]{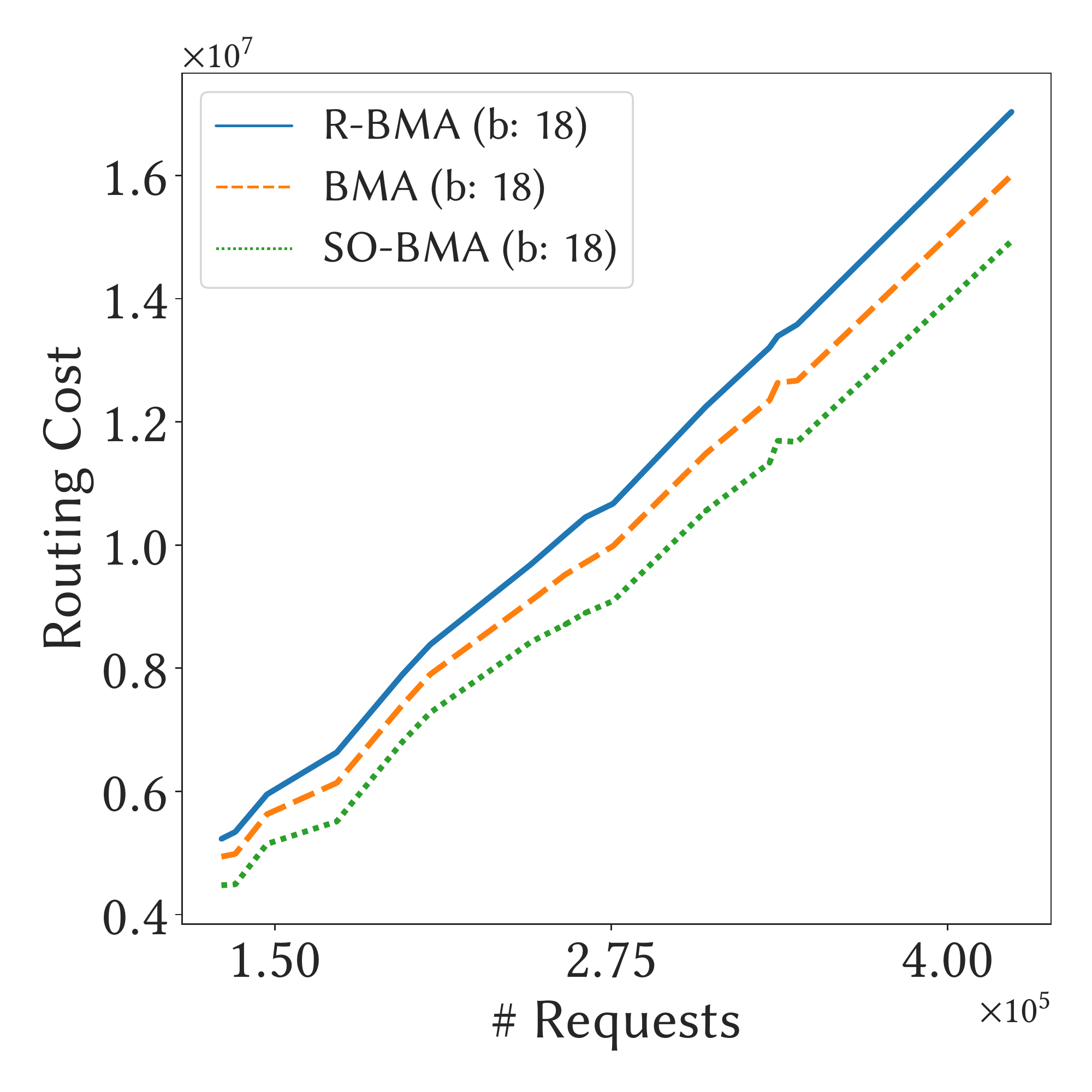}
		\caption{Best of comparison.}
		\label{subfig:fb-webservice-results-best}
\end{subfigure}
\caption{Facebook Web Service cluster.}
\label{fig:fb-webservice-results}
\end{figure*}

\begin{figure*}[!htb]
\centering
\begin{subfigure}[b]{0.325\textwidth}
		\centering
		\includegraphics[width=0.99\textwidth]{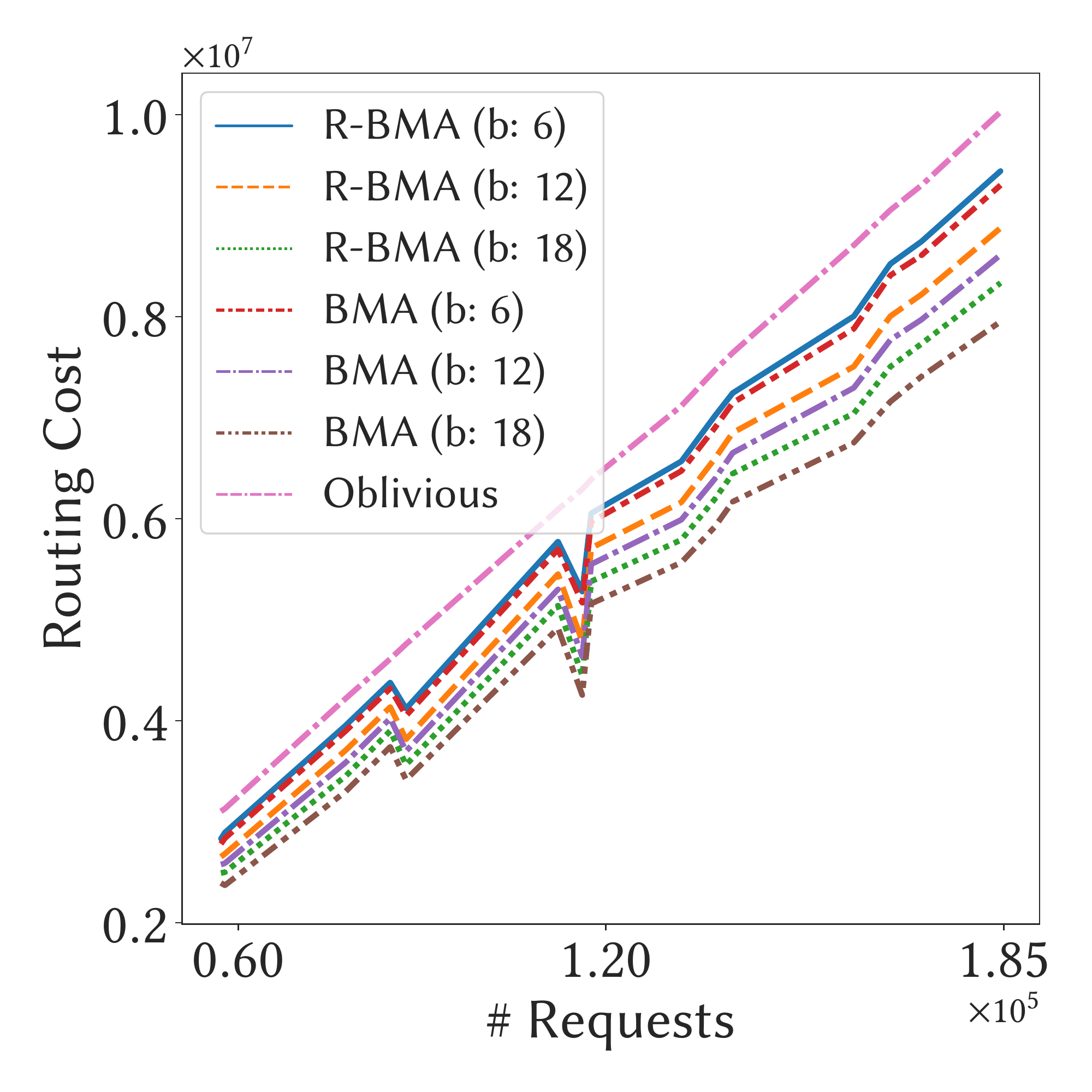}
		\caption{Routing costs.}
		\label{subfig:fb-hadoop-results-obj}
\end{subfigure}
\hfill
\begin{subfigure}[b]{0.325\textwidth}
		\centering
		\includegraphics[width=0.99\textwidth]{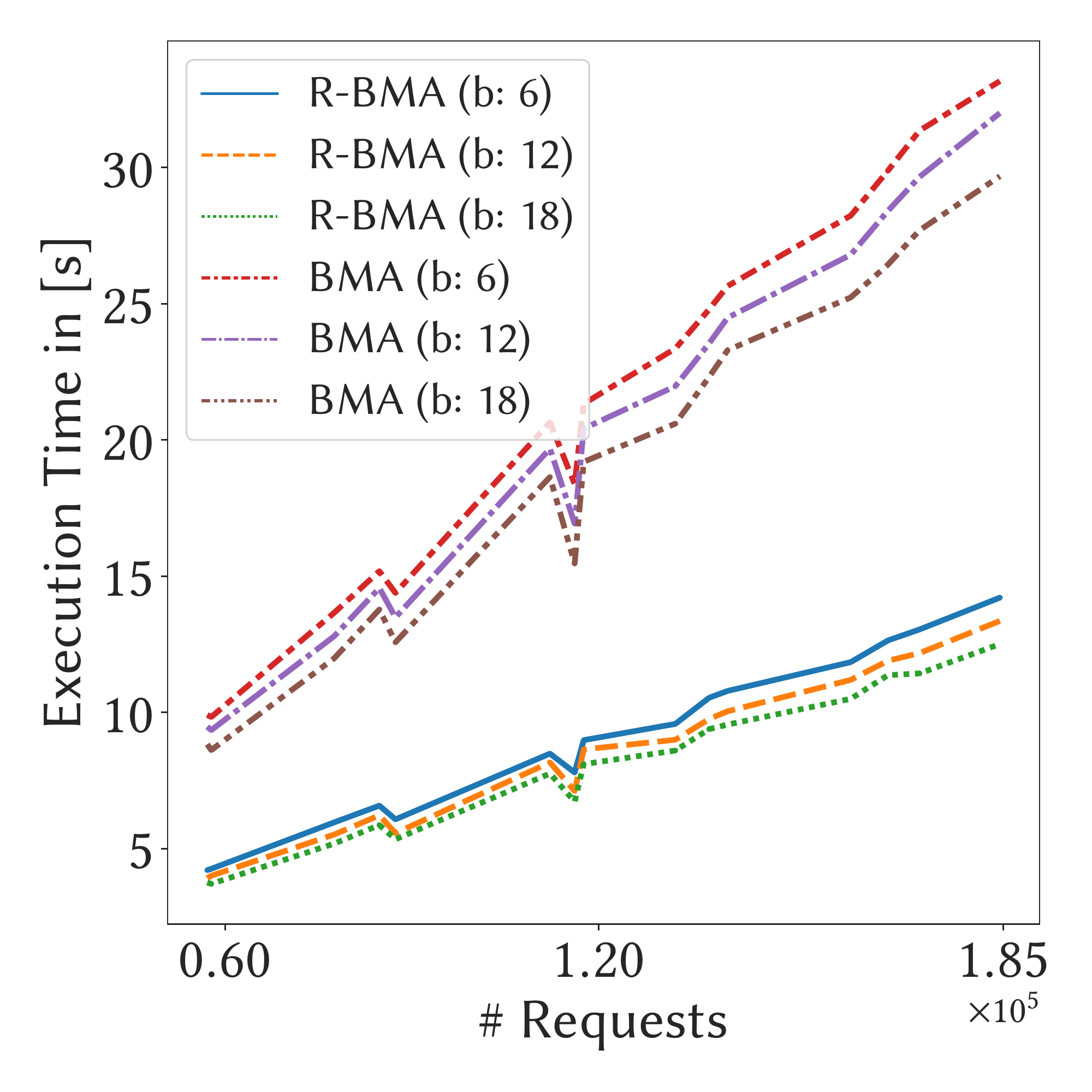}
		\caption{Execution time.}
		\label{subfig:fb-hadoop-results-exe}
\end{subfigure}
\hfill
\begin{subfigure}[b]{0.325\textwidth}
		\centering
		\includegraphics[width=0.99\textwidth]{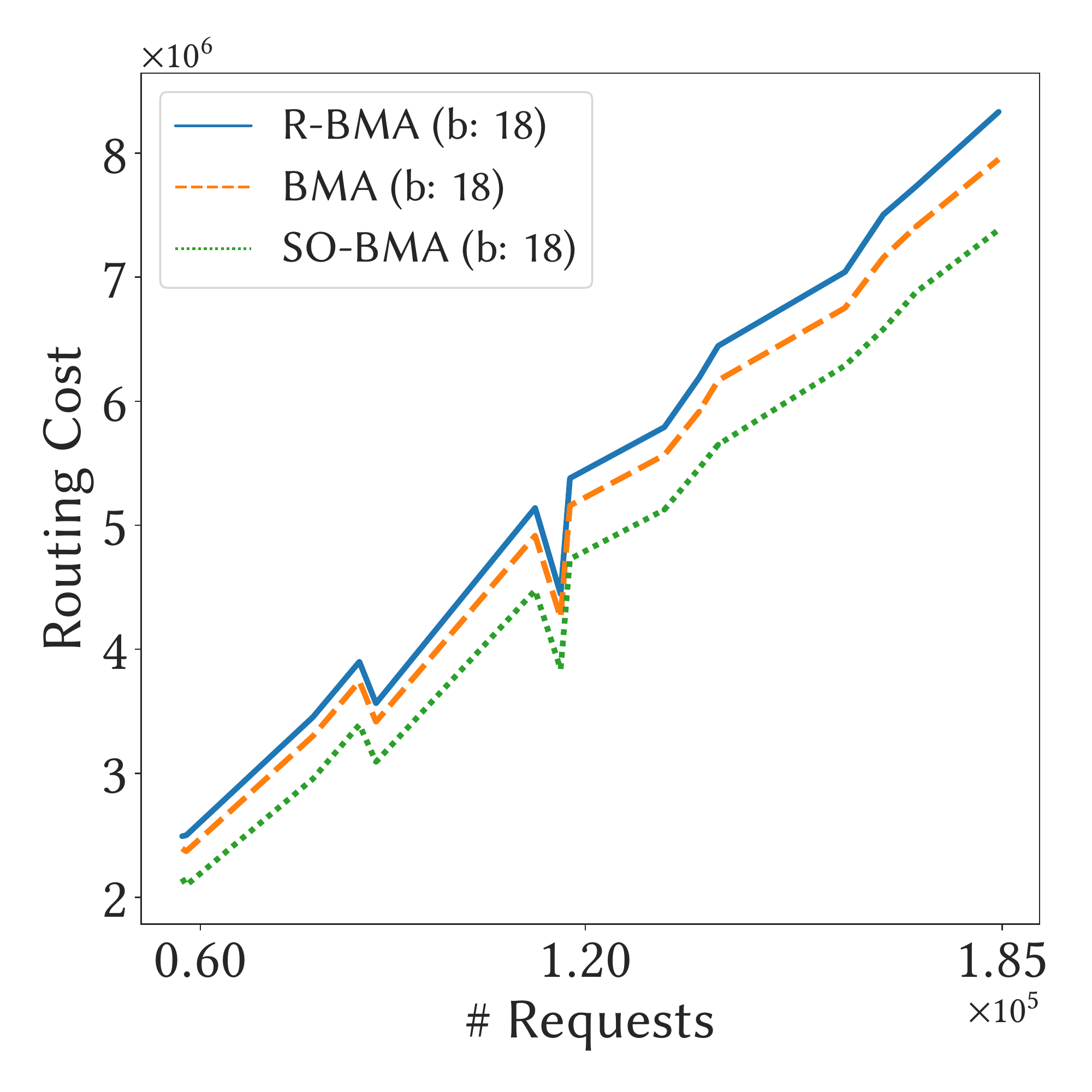}
		\caption{Best of comparison.}
		\label{subfig:fb-hadoop-results-best}
\end{subfigure}
\caption{Facebook Hadoop cluster.}
\label{fig:fb-hadoop-results}
\end{figure*}

\begin{figure*}[!htb]
\centering
\begin{subfigure}[b]{0.325\textwidth}
		\centering
		\includegraphics[width=0.99\textwidth]{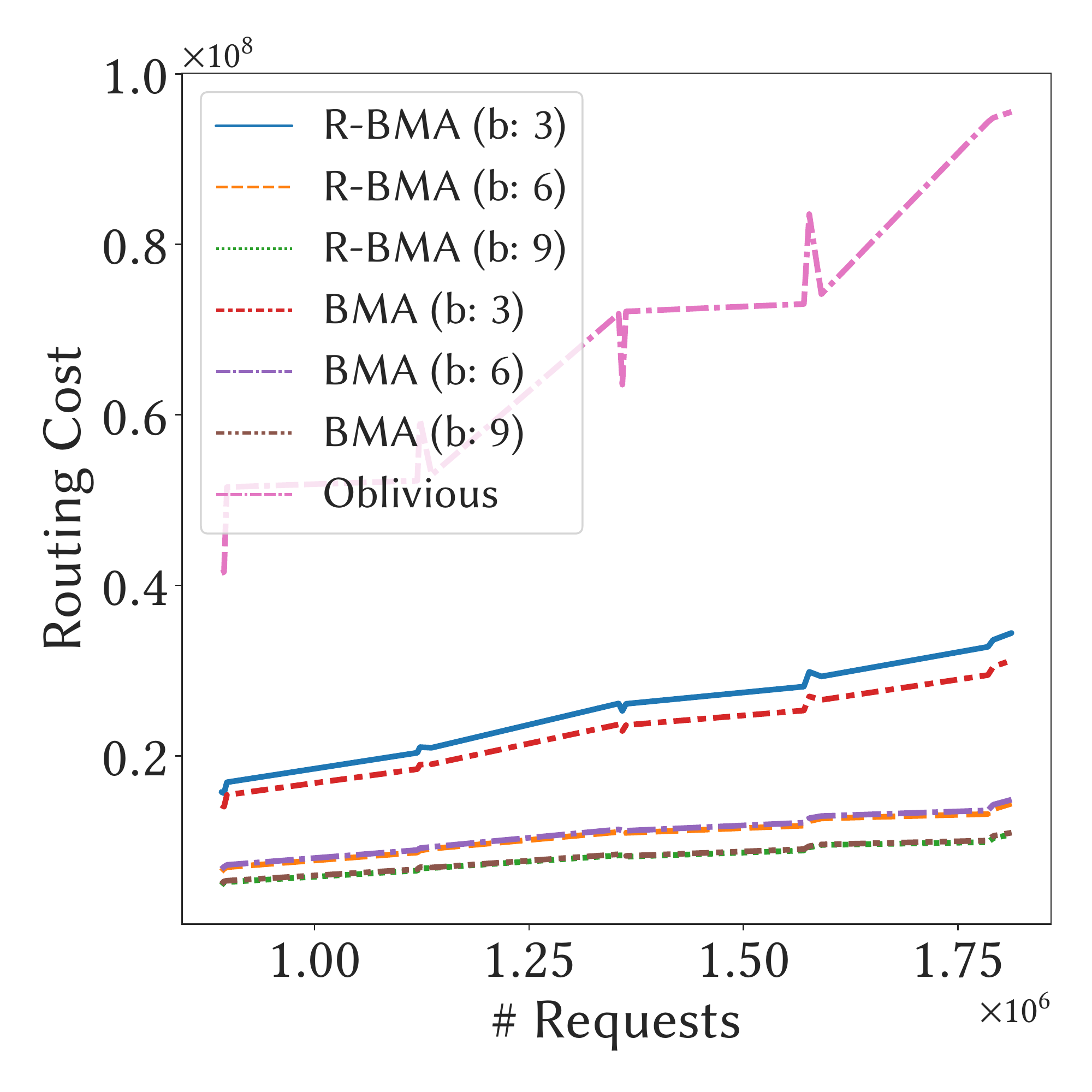}
		\caption{Routing costs.}
		\label{subfig:microsoft-results-obj}
\end{subfigure}
\hfill
\begin{subfigure}[b]{0.325\textwidth}
		\centering
		\includegraphics[width=0.99\textwidth]{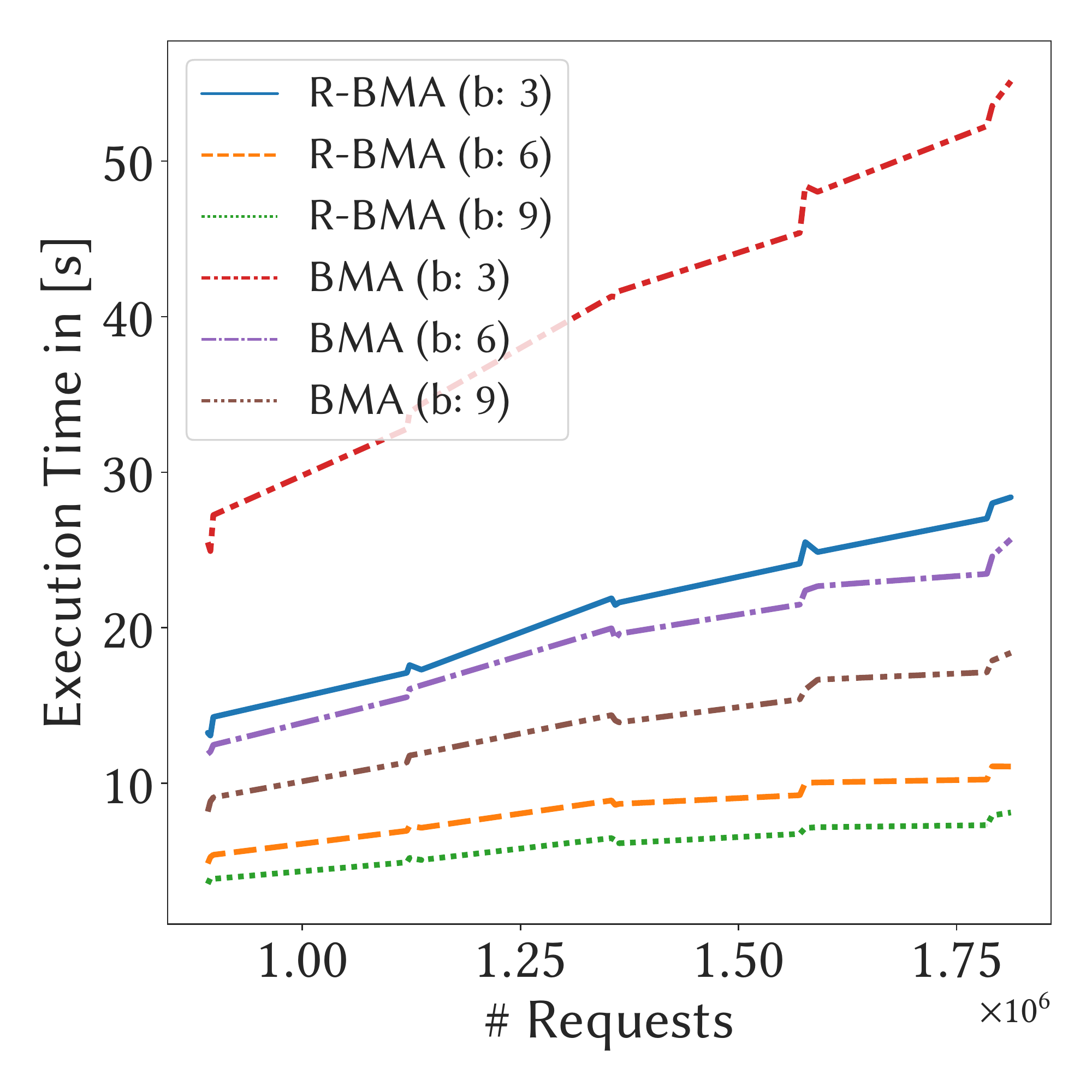}
		\caption{Execution time.}
		\label{subfig:microsoft-results-exe}
\end{subfigure}
\hfill
\begin{subfigure}[b]{0.325\textwidth}
		\centering
		\includegraphics[width=0.99\textwidth]{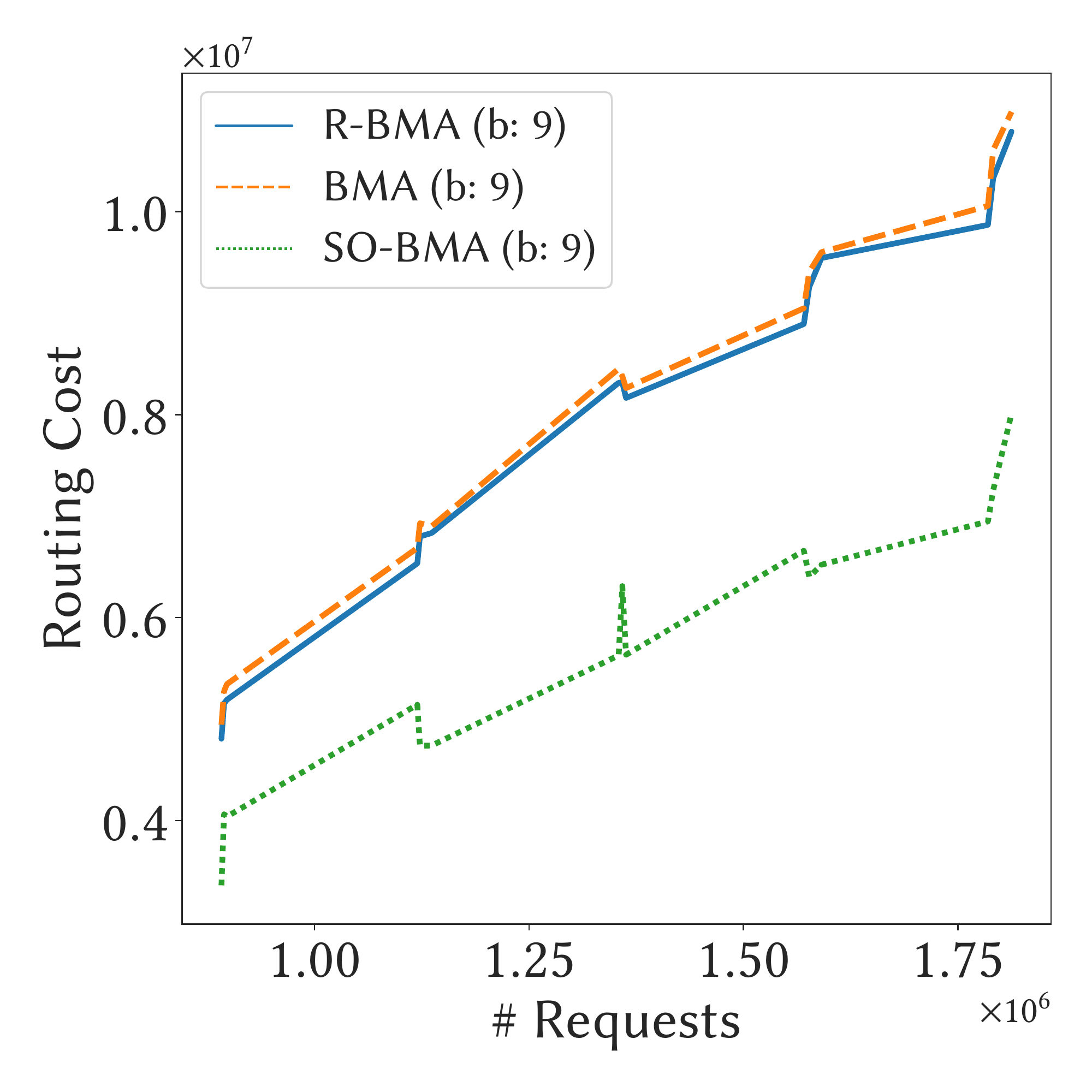}
		\caption{Best of comparison.}
		\label{subfig:microsoft-results-best}
\end{subfigure}
\caption{Microsoft cluster.}
\label{fig:microsoft-results}
\end{figure*}

\section{Empirical Evaluation}\label{sec:simulations}

To complement our formal evaluation, and in particular the worst-case analysis above, 
we also performed an empirical study of the performance of our algorithm under 
real-world workloads from various datacenter operators.
In particular, we benchmark our randomized algorithm against a state-of-the-art (deterministic) online b-matching algorithm~\cite{perf20bmatch} and against a maximum weight matching algorithm~\cite{maxmatching86}.

\subsection{Methodology}\label{sec:methodology}

\paragraph{Setup.} 

We implemented all algorithms in Python leveraging the
NetworkX library. 
For the implementation of the \emph{Maximum Weight Matching} algorithm we used the algorithm provided by NetworkX, which is based on Edmond's \emph{Blossom} algorithm~\cite{maxmatching86}.
Our experiments were run on a machine with two Intel Xeons
E5-2697V3 SR1XF with 2.6 GHz, 14 cores each and 128 GB
RAM. The host machine was running Ubuntu 20.04 LTS.

\paragraph{Simulation Workloads.}

Real-world datacenter traffic can vary significantly with respect to the spatial and temporal structure they feature, which depends on the application running \cite{sigmetrics20complexity}. 
Hence, our simulations are based on the following real-world datacenter traffic workloads from Facebook and Microsoft, which cover a spectrum of application domains.

\begin{itemize}
	\item \textbf{Facebook}~\cite{roy2015inside}:
    We use three different workloads, each from a different
    Facebook cluster. 
    We use a batch processing trace
	from one of Facebook's \emph{Hadoop} clusters, as well as traces
	from one of Facebook's database clusters, which serves SQL requests.
	Furthermore, we use traces from one of Facebook's \emph{Web-Service} cluster.   
	
	\item \textbf{Microsoft}~\cite{projector}:
	This data set is simply a probability
	distribution, describing the rack-to-rack
	communication (a traffic matrix).
	In order to generate a trace, we sample
	from this distribution \emph{i.i.d.}
	Hence, this trace does not
	contain any temporal structure by design
	(e.g., is not bursty)~\cite{sigmetrics20complexity}.
	However, it is known that it contains
	significant spatial structure (i.e.,
	is skewed). 
	
\end{itemize}

In all our simulations, we consider a typical fat-tree based datacenter topology, with~$100$ nodes in the case of the Facebook clusters, and with~$50$ nodes in the case of the Microsoft cluster.
Although our proposed algorithm can be used on any other static network, e.g., a~star topology, our experiments only consider the fat-tree topology because of its wide adoption across datacenters.

The main difference with respect to different network topologies is the cost associated with routing a request over the static network. 
Hence, network topologies with shorter paths between source and destination nodes would result in lower costs for all requests routed over the static network.

The cost of each request is calculated as the shortest path length between source and destination node.
Hence, if source and destination are connected by a reconfigurable link the cost equals~$1$.
Otherwise, the routing cost is computed as the number of hops from source to destination.
In general, each simulation is repeated~five times and then the results are averaged.
Each simulation is run sequentially. Hence, no parallelism is used during the execution of the proposed algorithm.

\subsection{Results and discussion}

We discuss the main results of our simulations based on the traces introduced above.
For each traffic trace we evaluate the routing cost and execution performance of our randomized $b$-matching algorithm. 
In particular, we evaluate the impact of different $b$ (henceforth called cache size due to how these links are managed in our algorithm) and compare our randomized algorithm (R-BMA) to the performance of BMA \cite{perf20bmatch} (BMA) and a \emph{Maximum Weight Matching} algorithm (SO-BMA). 

\paragraph{Routing Cost.} 

We first discuss the observed routing cost. 
Fig.~\ref{subfig:fb-database-results-obj} shows the results of the Facebook database cluster.
The violet line denotes the oblivious case, where each request is solely routed over the static network and no matching edges are present, e.g., a network without any reconfigurable switch.
The results show that R-BMA achieves a significant routing cost reduction of up to~$35\%$ with a cache size of~$18$.
In comparison to BMA, R-BMA performs almost identical to BMA on smaller cache sizes of~$6$.
Fig.~\ref{subfig:fb-database-results-best} shows that R-BMA routing cost reduction is not more than~$5\%$ higher on smaller numbers of requests, e.g., up to 200,000 requests.
Still, in comparison to the SO-BMA, the gap with respect to the routing cost reduction widens as the number of requests grows. 
In contrast to the results achieved on Facebook's database cluster, Figs.~\ref{subfig:fb-webservice-results-best} and~\ref{subfig:fb-hadoop-results-best} show that R-BMA  achieves similar routing cost reductions compared to BMA and SO-BMA.

Regarding the Microsoft traces, we can observe that R-BMA achieves a similar routing cost reduction compared to BMA. 
Furthermore, as the cache size grows, R-BMA achieves the same routing cost reduction as BMA, as shown in Fig.~\ref{subfig:microsoft-results-obj}.
Fig.~\ref{subfig:microsoft-results-best} shows that SO-BMA performs significantly better than R-BMA.
However, Microsoft's traces do not feature temporal structure and therefore offline algorithms such as SO-BMA have a significant advantage.

\paragraph{Execution Time.} 

All our simulations on all different traces show that our R-BMA algorithm 
outperforms BMA~\cite{perf20bmatch} with respect to run-time efficiency.
Furthermore, the size of the cache has a smaller impact on the execution time than BMA.
In particular, Figs.~\ref{subfig:fb-database-results-exe},  ~\ref{subfig:fb-webservice-results-exe}, ~\ref{subfig:fb-hadoop-results-exe} show that a larger cache size can lead to a decrease in execution performance of up to~$20\%$ in the case of the BMA algorithm, whereas our randomized algorithm is comparatively more robust to a change in the size of the cache.
The results of the Microsoft cluster in Fig.~\ref{subfig:microsoft-results-exe} also show that the run-time of 
our randomized algorithm grows slower than for BMA.

\paragraph{Summary.} 

Our R-BMA algorithm achieves almost the same routing cost reduction as BMA, while also achieving competitive cost reductions compared to an optimal offline algorithm.
With respect to the run-time efficiency of our randomized algorithm, we find that our algorithm significantly outperforms BMA on all workloads.
We conclude that our algorithm, R-BMA, provides an attractive trade-off between routing cost reduction and run-time efficiency. 

\section{Related Work}\label{sec:relwork}

The design of datacenter topologies has received much attention
in the networking community already. The most widely deployed
networks are based on Clos topologies and multi-rooted fat-trees~\cite{clos,singh2015jupiter,f10},
and there are also interesting designs based on hypercubes~\cite{bcube,mdcube} and expander
graphs~\cite{xpander,jellyfish}.

Existing dynamic and demand-aware datacenter networks can be classified according to the
granularity of reconfigurations. 
Solutions such as Proteus~\cite{proteus}, OSA~\cite{osa}, 
 or DANs~\cite{dan}, among other, are more coarse-granular
and e.g., rely on a (predicted) traffic matrix.
Solutions such as  ProjecToR~\cite{projector,spaa21rdcn}, Cerberus~\cite{griner2021cerberus}, 
MegaSwitch~\cite{megaswitch}, Eclipse \cite{venkatakrishnan2018costly},
Helios~\cite{helios}, Mordia~\cite{mordia}, Duo~\cite{zerwas2023duo}, 
C-Through~\cite{cthrough}, ReNet~\cite{apocs21renets} or SplayNets~\cite{splaynet}
are more fine-granular and support per-flow reconfiguration
and decentralized reconfigurations.
Reconfigurable demand-aware networks may also rely on expander graphs, e.g., Flexspander~\cite{flexspander}, Cerberus~\cite{griner2021cerberus}, or
Duo~\cite{zerwas2023duo}, and are currently also considered as a promising solution to speed up data transfers in supercomputers~\cite{100times,fleet}.
The notion of demand-aware networks raise novel optimization problems related to switch scheduling~\cite{mckeown1999islip}, and recently interesting first insights have been obtained both for offline~\cite{venkatakrishnan2018costly} and for online scheduling~
\cite{schwartz2019online,dinitz2020scheduling,spaa21rdcn,perf20bmatch}.
Due to the increased reconfiguration time 
experienced in demand-aware networks,
many existing demand-aware architectures additionally
rely on a fixed network. 
For example, 
ProjecToR always maintains a
``base mesh'' of connected links that can handle low-latency
traffic while it opportunistically reconfigures free-space links
in response to changes in traffic patterns. 

This paper primarily focuses on the \emph{algorithmic}
problems of demand-aware datacenter architectures. 
Our optimization problem is related to \emph{graph augmentation}
models, which consider the problem of adding edges to a given
graph, so that path lengths are reduced. For example, 
Meyerson and Tagiku~\cite{meyerson2009minimizing} study how to add ``shortcut edges''
to minimize the average shortest path distances, 
Bil{\`o} et al.~\cite{bilo2012improved} and
Demaine and Zadimoghaddam~\cite{demaine2010minimizing}
study how to augment a~network to reduce its diameter,
and there are several interesting results on how to add ``ghost edges'' 
to a graph such that it becomes (more) ``small world''~\cite{ghost-edges,small-world-shortcut,gozzard2018converting}.
However, these edge additions can be optimized globally and in a biased manner, and
hence do not form a matching; we are also not aware of any online versions of this problem. 
The dynamic setting is related to classic switch scheduling
problems~\cite{mckeown1999islip,chuang1999matching}.

Regarding the specific $b$-matching problem considered in this paper, a polynomial-time algorithm for the static 
version of this problem is known for several decades~\cite{Schrij03,anstee1987polynomial}.
Recently, Hanauer et al.~\cite{infocom22matching,infocom23matching} presented several efficient algorithms for the static~\cite{infocom22matching} and dynamic (but offline)~\cite{infocom23matching} problem variant for application in the context of reconfigurable datacenters. 
Bienkowski et al.~\cite{perf20bmatch} initiated the study of an online version of this problem
and presented a $O(b)$-competitive
deterministic algorithm and showed that this is asymptotically optimal.
In this paper, we have shown that a randomized approach can provide a significantly lower competitive
ratio as well as faster runtimes.

Finally, we note that there is a line of papers studying 
(bipartite) online matching
variants
\cite{online-matching,online-matching-simple,adwords-primal-dual,concave-matching,ranking-primal-dual,bipartite-matching-strongly-lp,adwords-lp,adwords-ec}.
This problem attracted significant attention in the last decade because of its connection
to online auctions and the AdWords problem~\cite{adwords-survey}. 
Despite similarity in names (e.g., the
bipartite (static) $b$-matching variant was considered
in~\cite{kalyanasundaram2000optimal}), this model is fundamentally
different from ours.

\section{Conclusion}\label{sec:conclusion}

We revisited the problem of how to schedule reconfigurable links
in a datacenter (based on optical circuit switches or similar technologies),
in order to maximize network utilization. 
To this end, we presented a randomized online algorithm which computes 
heavy matchings between, e.g., datacenter racks,
guaranteeing a significantly lower competitive ratio and faster running
time compared to the state-of-the-art (and asymptotically optimal) 
deterministic algorithm. Our algorithm and its analysis are simple,
and easy to implement (and teach).

That said, our work leaves open several interesting avenues for future research.
In particular, we still lack a non-asymptotic and tight bound on the achievable competitive ratio
both in the deterministic and in the randomized case.
Furthermore, we so far assumed a conservative online perspective, where the algorithm does not have
any information about future requests. In practice, traffic often features temporal structure,
and it would be interesting to explore algorithms which can leverage certain predictions about
future demands, without losing the worst-case guarantees.

\bibliographystyle{plainurl}
\bibliography{refs}

\end{document}